\documentclass[9pt,a4paper]{article}

\usepackage{titlesec}
\usepackage{amssymb}
\usepackage{amsmath,amsthm,amssymb,color}
\usepackage{graphicx}
\date{}
\usepackage{url}
\urldef{\mailsa}\path|{alfred.hofmann, ursula.barth, ingrid.haas, frank.holzwarth,|
\urldef{\mailsb}\path|anna.kramer, leonie.kunz, christine.reiss, nicole.sator,|
\urldef{\mailsc}\path|erika.siebert-cole, peter.strasser, lncs}@springer.com|

\newtheorem{lem}{Lemma}[section]
\newtheorem{thm}[lem]{Theorem}
\newtheorem{prop}[lem]{Proposition}
\newtheorem{coro}[lem]{Corollary}
\newtheorem{defn}[lem]{Definition}

\newtheorem{rmk}[lem]{Remark}

\begin{document}
\title{Constructing vectorial bent functions via second-order derivatives }

\author{Lijing Zheng~$^{\diamond}$, Jie Peng~$^{*}$, Haibin Kan~$^{\diamond }$, Yanjun Li
\thanks{ J. Peng and Y. Li are with Mathematics and Science College of
Shanghai Normal University, Guilin Road \#100, Shanghai, China, 200234, (E-mails:~jpeng@shnu.edu.cn,~yanjlmath90@163.com).
\newline \indent $^{\diamond}$ L. Zheng and H. Kan are with the School of Computer Sciences, Fudan University, Handan Road \#220,
Shanghai, 200433, China; L. Zheng is also with the School of Mathematics and Physics, University of South China, Changsheng Road \#28, Hengyang, Hunan, 421001, China, (E-mails:~zhenglijing817@163.com, hbkan@fudan.edu.cn).}}

\maketitle

\begin{abstract}
Let $n$ be an even positive integer, and $m<n$ be one of its positive divisors. In this paper, inspired by a nice work of Tang et al. on constructing large classes of bent functions from known bent functions [27, IEEE TIT, 63(10): 6149-6157, 2017], we consider the construction of vectorial bent and vectorial plateaued $(n,m)$-functions of the form $H(x)=G(x)+g(x)$, where $G(x)$ is a vectorial bent $(n,m)$-function, and $g(x)$ is a Boolean function over $\mathbb{F}_{2^{n}}$. We find an efficient generic method to construct vectorial bent and vectorial plateaued functions of this form by establishing a link between the condition on the second-order derivatives and the key condition given by [27]. This allows us to provide (at least) three new infinite families of vectorial bent functions with high algebraic degrees. New vectorial plateaued $(n,m+t)$-functions are also obtained ($t\geq 0$ depending on $n$ can be taken as a very large number), two classes of which have the maximal number of bent components.

\end{abstract}
\medskip

\section{Introduction}
Throughout this paper, we often identify the finite field $\mathbb{F}_{2^{n}}$ with $\mathbb{F}^{n}_{2}$, the $n$-dimensional vector space over $\mathbb{F}_{2}$. Any function $F: \mathbb{F}_{2^{n}}\rightarrow\mathbb{F}_{2^{m}}$ is called an {\it $(n,m)$-function}, which is also called a {\it Boolean function} when $m=1$.
{\it Bent functions} have been introduced by Rothaus in 1976 \cite{Rothaus}. They are Boolean functions in even number of variables which are maximally nonlinear in the sense that their Hamming distance to all affine Boolean functions is optimal. It corresponds to the fact that the {\it Walsh transform} of a bent function in $n$ variables takes precisely the values $\pm 2^{\frac{n}{2}}$. Over the last four decades, bent functions have attracted a lot of research interest because of their applications in coding theory, combinatorics and cryptography. A survey on bent functions can be found in \cite{Car-Mes16}, as well as the book \cite{Mes16}.

The bent property of Boolean functions has been extended to general $(n,m)$-functions $F$ by requesting that all the nonzero linear combinations of its {\it coordinate functions} are bent functions. Such vectorial functions are called {\it vectorial bent}. They exist if and only if $n$ is even and $m\leq n/2$.
In the literature, methods to construct new vectorial bent (plateaued) functions are divided into two classes: those building functions from scratch are called {\it primary};  those using known vectorial bent functions are called {\it secondary}. For primary constructions, Nyberg firstly presented the constructions of vectorial bent functions based on some special classes of bent functions such as the Maiorana-McFarland class ($\mathcal{MM}$), the Dillon's partial spread class ($\mathcal{PS}$) and class $H$ \cite{Nyb} (this class has been generalized to class $\mathcal{H}$ by Carlet and Mesnager, see \cite{Car-Mes11}). Satoh et al. modified the first method in \cite{Nyb} such that the constructed functions achieve the largest algebraic degree. Pasalic and Zhang studied vectorial bent functions of the form  $F(x)={\rm Tr}^{n}_{m}(\lambda x^d)$ \cite{PZ12}. Dong et al. constructed three classes of vectorial bent $(2k,k)$-functions based on monomial bent functions and $\mathcal{PS}$ bent functions, see \cite{DZQF13}. Muratovi\'{c}-Ribi\'{c} et al. studied the vectorial bentness and hyperbentness of the trace functions ${\rm Tr}^{2k}_{k}(\sum\limits^{2^{k}}_{i=0}\alpha_{i}x^{i(2^{k}-1)})$, $\alpha_{i}\in \mathbb{F}_{2^{n}}$, see \cite{MPS, MPS14,PTK17}. Mesnager presented a generic construction of bent vectorial $(2k,k)$-functions of the form $xG(yx^{2^{k}-2})$, where $G$ is an oval polynomial on $\mathbb{F}_{2^{k}}$ \cite{Mes15}. Xu et al. gave a classification of vectorial bent monomials and some constructions of bent multinomials in \cite{Xu18}.

Compared with primary constructions, the results on secondary constructions of vectorial bent functions seem to be much fewer. Carlet and Mesnager proposed some new secondary constructions of vectorial bent functions with larger numbers of  variables \cite{CM10}. In \cite{Bu-Car11}, Budaghyan and Carlet showed that two {\it CCZ-equivalent} vectorial bent functions must be {\it EA-equivalent}, and hence have the same algebraic degree. Further, the authors gave a method to produce non-quadratic vectorial bent functions by applying CCZ-equivalence to non-bent vectorial functions which have some components of bent functions. Very recently, Pott et al. proved that an $(n,n)$-function can have at most $2^{n}-2^{\frac{n}{2}}$ bent components, and those possess the maximum number of components can produce new vectorial bent functions \cite{Pott18}. More precisely, let $n=2k$, for any $(n,n)$-function $G$, if ${\rm Tr}^{n}_{1}(\alpha G(x))$ is bent for any $\alpha\in\mathbb{F}_{2^{n}}\setminus\mathbb{F}_{2^{k}}$, then ${\rm Tr}^{n}_{k}(\alpha G(x))$ is vectorial bent for any $\alpha \in \mathbb{F}_{2^{n}}\setminus\mathbb{F}_{2^{k}}$. Based on this observation, they find an infinite class of (quadratic) bent $(n,\frac{n}{2})$-functions of the form ${\rm Tr}^{n}_{k}(\alpha x^{2^{i}}({x+x^{2^{k}}}))$, where $\alpha\in\mathbb{F}_{2^{n}}\setminus\mathbb{F}_{2^{k}}$. In \cite{ZP18}, the authors show that for an $(n,m)$-function $G$ with $m\geq \frac{n}{2}$, the maximal possible number of bent components is equal to $2^{m}-2^{m-\frac{n}{2}}$, and those with maximum number of bent components can also produce optimal vectorial bent functions. They found a generic  class of  bent $(n,\frac{n}{2})$-functions of the form ${\rm Tr}^{n}_{k}(\alpha x^{2^{i}}\pi({x+x^{2^{k}}}))$, where $\pi$ is a permutation over $\mathbb{F}_{2^{k}}$, and  $\alpha\in\mathbb{F}_{2^{n}}\setminus\mathbb{F}_{2^{k}}$. This is why in this article we concern not only vectorial bent functions but also vectorial functions with maximal number of bent components.

In this paper, however, we mainly focus on the secondary constructions of vectorial bent functions without increasing the number of variables and then try to utilize those resulting vectorial bent functions to generate vectorial (plateaued) functions with maximal number of bent components. Explicitly, for an even integer $n=2k$, and its positive divisor $m$ such that $m\leq n/2$, we consider vectorial $(n,m)$-functions of the form
\begin{eqnarray}\label{cons1}
H(x)=G(x)+g(x),
 \end{eqnarray} where $G(x)$ is a vectorial bent $(n,m)$-function, and $g(x)$ is a Boolean function over $\mathbb{F}_{2^{n}}$.

At first glace, it would appear that finding such functions $G(x)$, and $g(x)$ might be quite difficult. However,
this is in fact not the case. By observing recent nice works of Mesnager \cite{Mes14}, and Tang et al. \cite{Tang17} on constructing of bent Boolean functions, in this paper we firstly introduce a {\it property} ($\mathbf{P}_{\tau}$) concerning Boolean functions and then establish a link between this property and the condition of Construction 7 presented by Tang et al \cite{Tang17}. This powerful tool makes us efficiently find more bent functions and then we find (at least) three new infinite families of vectorial bent functions by choosing some specific classes of vectorial (bent) functions.
It turns out that for each class of the selected vectorial (bent) functions, there are many $g's$ satisfying the required conditions. This also makes it possible for us to further construct {\it vectorial (plateaued) } $(n,m+t)$-functions which have the maximal number of bent components in the sense of \cite{ZP18}, see also \cite{Pott18} for the special case of $(n,n)$-functions, where $t$ is a nonnegative integer.

The rest of the paper is organized as follows. In Section 2 some basic definitions are given. In Section 3, based on the works of Carlet and Mesnager, we introduce the definition of property $(\mathbf{P}_{\tau})$, and establish a link between  property $(\mathbf{P}_{\tau})$ and the condition of Construction 7 of \cite{Tang17}. In Section 4, we specify how to produce new vectorial bent (plateaued) functions of the form \eqref{cons1}. In Section 5, we show that the results obtained in Section 4 give rise to (at least) three new infinite families of vectorial bent functions, and vectorial plateaued functions which have the maximal bent number of bent components. Finally, we concludes this paper in Section 6.

\section{Preliminaries}

Let $\mathbb{F}_{2^{n}}$ be the finite field consisting of $2^{n}$ elements, then its multiplicative group, denoted by $\mathbb{F}^{\ast}_{2^{n}}$, is a cyclic group of order $2^{n}-1$. Throughout this paper, we always identify $\mathbb{F}_{2^{n}}$ with the vector space $\mathbb{F}^{n}_{2}$ over $\mathbb{F}_{2}$. Any function $F: \mathbb{F}_{2^{n}}\rightarrow\mathbb{F}_{2^{m}}$ is called an $(n,m)$-function. Usually $(n,1)$-functions are called Boolean functions in $n$ variables, the set of which is denoted by $\mathcal{B}_{n}$.

The trace function ${\rm Tr}^{n}_{m}: \mathbb{F}_{2^{n}}\rightarrow \mathbb{F}_{2^{m}}$, where $m~|~n$, is defined as
$${\rm Tr}^{n}_{m}(x)\!=\!x\!+\!x^{2^{m}}\!+\!x^{2^{2m}}\!+\!\cdots\!+\!x^{2^{(n/m\!-\!1)m}}, ~\forall ~x\in \mathbb{F}_{2^{n}}.$$ When $m=1$, it is also called the absolute trace. In this paper, $\langle,\rangle$ denotes the usual {\it inner product} in a vector space over $\mathbb{F}_{2}$. For any $\alpha=(\alpha_1,\ldots,\alpha_n), \beta=(\beta_1,\ldots,\beta_n)\in\mathbb{F}_2^{n}$, one has $\langle\alpha,\beta\rangle=\sum_{i=1}^n\alpha_i\beta_i$. While in the finite field $\mathbb{F}_{2^{n}}$, we take $\langle\alpha,\beta\rangle={\rm Tr}^{n}_{1}(\alpha\beta)$ for any $\alpha,\beta\in\mathbb{F}_{2^{n}}$.

 For any $(n,m)$-function $F=(f_1,\ldots,f_m)$, where $f_1,\ldots,f_m\in\mathcal{B}_{n}$, all the nonzero linear combinations of
 $f_i, 1\leq i\leq m$, are
called the components of $F$. When $F$ is viewed as a mapping from the finite field $\mathbb{F}_{2^{n}}$ to $\mathbb{F}_{2^{m}}$,
the components of $F$ can be represented as $F_{\lambda}(x)={\rm Tr}_1^m(\lambda F(x)), ~\lambda\in\mathbb{F}_{2^{m}}^*.$

A Boolean function $f\in\mathcal{B}_{n}$ can
be uniquely represented by a multivariate polynomial as
\begin{eqnarray}\label{boolean}
f(X_1,\ldots,X_n)\!=\!\sum_{I\subseteq
\{1,2,\ldots,n\}}a_I\prod_{i\in I}X_i,~a_{I}\in \mathbb{F}_{2}.
\end{eqnarray}
A polynomial in $\mathbb{F}_{2}[X_1, \ldots, X_n]$ of the form (\ref{boolean}) is called a {\it reduced polynomial}. The number of variables in the highest order term with nonzero coefficient of this polynomial is called the {\it algebraic degree} of $f$. While for a general $(n,m)$-function $F$, the highest algebraic degree of its coordinate functions is called the algebraic degree of $F$. The function $F$ is called {\it quadratic} if its algebraic degree is no more than 2.

The {\it Walsh transform} of a Boolean function $f\in \mathcal{B}_{n}$ at a point $a \in \mathbb{F}_{2^{n}}$ is defined by
$$W_{f}(a)\!=\!\sum\limits_{x\in \mathbb{F}_{2^{n}}}(\!-\!1)^{f(x)\!+\!{\rm Tr}^{n}_{1}(a x)}.$$

 The function $f$ is called bent if $|W_{f}(a)|=2^{\frac{n}{2}}$ for all $a \in \mathbb{F}_{2^{n}}$.
 It is well known that bent functions exist if and only if $n$ is even. When $f\in \mathcal{B}_{n}$ is bent, the Boolean function $f^{\ast}$ such that $W_{f}(\alpha)=2^{\frac{n}{2}}(-1)^{f^{\ast}(\alpha)}$ for any $\alpha\in \mathbb{F}_{2^{n}}$, is also bent and is called the dual of $f$.

A Boolean function $f$ is called {\it plateaued} if $W_f$ takes three values $\{0,\pm 2^{s}\}$ for some integer $n/2\leq s\leq n$. For the case of $n$ even, the function $f$ is called {\it semi-bent} if $W_f$ takes three values $\{0,\pm 2^{\frac{n}{2}+1}\}$.

The nonlinearity of an $(n,m)$-function $F$ and hereby its resistance to {\it linear cryptanalysis} \cite{Matsui} is measured through the {\it extended Walsh spectrum}
$$\{\ast |W_{F}(a,\gamma)|:\gamma \in \mathbb{F}^{m}_{2}\backslash\{0\}, a \in \mathbb{F}_{2^{n}} \ast\},$$
where
$$W_{F}(a,\gamma)\!=\!\sum\limits_{x\in  \mathbb{F}^{n}_{2}}(\!-\!1)^{\langle \gamma,F(x)\rangle\! +\!\langle a,x\rangle }.$$
The function $F$ is said to be a {\it vectorial~bent~function} of dimension $m$ if all the components of $F$ are bent. In other words, $F$ is vectorial bent if and only if $|W_{F}(a,\gamma)|=2^{n/2}$, for any $\gamma \in \mathbb{F}^{m}_{2}\backslash\{0\}$ and for any $a\in\mathbb{F}^{n}_{2}$. $F$ is said to be a {\it plateaued vectorial function} if all the components are plateaued Boolean functions.

Two $(n,m)$-functions $F$ and $G$ are called {\it extended affine equivalent} (EA-equivalent) if there exist some affine permutation $L_1$ over $\mathbb{F}_{2^n}$ and some affine permutation $L_2$ over $\mathbb{F}_{2^m}$, and some affine function $A$ such that $F=L_2\circ G\circ L_1+A$. They are called {\it Carlet-Charpin-Zinoviev equivalent} (CCZ-equivalent) if there exists some affine automorphism $L=(L_1, L_2)$ of $\mathbb{F}_{2^n}\times \mathbb{F}_{2^m}$, where $L_1: \mathbb{F}_{2^n}\times \mathbb{F}_{2^m}\rightarrow \mathbb{F}_{2^n}$ and $L_2: \mathbb{F}_{2^n}\times \mathbb{F}_{2^m}\rightarrow \mathbb{F}_{2^m}$ are affine functions, such that $y=G(x)$ if and only if $L_2(x, y)=F\circ L_1(x, y)$. It is well known that EA-equivalence is a special kind of CCZ-equivalence, and that  CCZ-equivalence preserves the extended Walsh spectrum and the differential spectrum (but not for algebraic degree) \cite{CCZ}.

\section{The introduction of property $(\mathbf{P}_{\tau})$}

Throughout this section, let $n,\tau$ be two positive integers. To produce vectorial bent functions of the form \eqref{cons1} is to find suitable functions $G(x)$ and $g(x)$. To this end, in this section we introduce the property $(\mathbf{P}_{\tau})$ and establish a link between this property and the condition of Construction 7 presented by Tang et al. in \cite{Tang17}. This tool will make us effectively find new infinite families of vectorial bent functions which will be presented in the following two sections.\\

\noindent {\bf A. Carlet-Mesnager's criterion}\\

In this subsection, we firstly recall some known results which are the motivation for introducing property ($\mathbf{P}_{\tau}$) and present some new results concerning bent functions. The following construction is due to Carlet which can generate new bent functions \cite[Theorem 3]{Carlet06} and new plateaued functions \cite[Proposition 2]{Carlet15}.

\begin{lem}\label{carlet-lemma} (\cite[Lemma~1]{Carlet06}) Let $n$ be a positive integer. Let $f_{1},f_{2}, f_{3}, f_{4}\in \mathcal{B}_{n}$ be Boolean functions such that $f_{1}+f_{2}+f_{3}+f_{4}=0$. Let $\sigma: \mathbb{F}_{2^{n}}\rightarrow\mathbb{F}_{2}$ be defined as  $\sigma(x)=f_{1}(x)f_{2}(x)+f_{1}(x)f_{3}(x)+f_{2}(x)f_{3}(x)$. Then for each $a\in \mathbb{F}_{2^{n}}$,
\begin{center}
$W_{\sigma}(a)=\frac{1}{2}(W_{f_{1}}(a)+W_{f_{2}}(a)+W_{f_{3}}(a)-W_{f_{4}}(a)).$
\end{center}
\end{lem}

Let $n$ be an even integer. With the notations as above, based on a work of Carlet~(\cite[Theorem 3]{Carlet06}),  Mesnager has shown that if $f_{i}$ is bent for $i=1,2,3,4$, then  $\sigma$ is bent if and only if $f^{\ast}_{1}+f^{\ast}_{2}+f^{\ast}_{3}+f^{\ast}_{4}=0$; and if $\sigma$ is bent, then $\sigma^{\ast}=f^{\ast}_{1}f^{\ast}_{2}+f^{\ast}_{1}f^{\ast}_{3}+f^{\ast}_{2}f^{\ast}_{3}$, see \cite[Theorem~4]{Mes14}. Under the assumptions as in Lemma \ref{carlet-lemma}, we call this method of estimating whether $\sigma$ is bent or not \emph{Carlet-Mesnager's criterion}. We need the proof of this theorem and let us recall it as follows~(there are some improvements).
\begin{thm}\label{Carlet-Mesnager}(Carlet-Mesnager's criterion) Let $n=2k$ be a positive integer. Let $f_{1},f_{2}, f_{3}\in \mathcal{B}_{n}$ be three pairwise distinct bent functions such that $f_{4}=f_{1}+f_{2}+f_{3}$ is also a bent function. Let $\sigma=f_{1}f_{2}+f_{1}f_{3}+f_{2}f_{3}$.
Then $\sigma$ is bent if and only if $f^{\ast}_{1}+f^{\ast}_{2}+f^{\ast}_{3}+f^{\ast}_{4}=0$. Moreover, if $\sigma$ is bent, then $\sigma^{\ast}=f^{\ast}_{1}f^{\ast}_{2}+f^{\ast}_{1}f^{\ast}_{3}+f^{\ast}_{2}f^{\ast}_{3}$.
\end{thm}
\begin{proof}
 By Lemma \ref{carlet-lemma}, for each $a\in \mathbb{F}_{2^{n}}$, we have
\begin{eqnarray*} W_{\sigma}(a) &\!=\!& 2^{k-1}((-1)^{f_{1}^{\ast}(a)}+(-1)^{f_{2}^{\ast}(a)}+(-1)^{f_{3}^{\ast}(a)}-(-1)^{f_{4}^{\ast}(a)})\\
             &\!=\!& 2^{k}(f_{1}^{\ast}(a)+f_{2}^{\ast}(a)+f_{3}^{\ast}(a)+f_{4}^{\ast}(a)+1)~({\rm mod}~2^{k+1}).
             \end{eqnarray*}
Recall that for a Boolean function $g$, $W_{g}(a)=\pm 2^{k}$ if and only if $W_{g}(a)=2^{k}~({\rm mod}~2^{k+1})$. Then
  \begin{eqnarray}\label{0-pla} W_{\sigma}(a)=\pm 2^{k} \text{~if and only if~} f^{\ast}_{1}(a)+f^{\ast}_{2}(a)+f^{\ast}_{3}(a)+f^{\ast}_{4}(a)\equiv0~({\rm mod}~2),
\end{eqnarray}
and one has $\sigma$ is bent if and only if $f^{\ast}_{1}+f^{\ast}_{2}+f^{\ast}_{3}+f^{\ast}_{4}=0$; and the second assertion follows from Theorem 3 of \cite{Carlet06}.
\end{proof}

 In fact, using some arguments of the proof discussed above, we can obtain the following result.

\begin{thm}\label{startingpoint} With the same notations as in Theorem \ref{Carlet-Mesnager}. Then the following three assertions hold:

1) $\sigma$ is bent if and only if $f^{\ast}_{1}+f^{\ast}_{2}+f^{\ast}_{3}+f^{\ast}_{4}=0$; and if $\sigma$ is bent, then $\sigma^{\ast}=f^{\ast}_{1}f^{\ast}_{2}+f^{\ast}_{1}f^{\ast}_{3}+f^{\ast}_{2}f^{\ast}_{3}$;

2) $\sigma$ is semi-bent if and only if $f^{\ast}_{1}+f^{\ast}_{2}+f^{\ast}_{3}+f^{\ast}_{4}=1$;

3) Otherwise, $\sigma$ is a Boolean function satisfying $\{|W_{\sigma}(\lambda)|~|~\lambda\in \mathbb{F}_{2^{n}}\}=\{0,2^{k},2^{k+1}\}$.
\end{thm}
\begin{proof} One has seen that the first assertion holds true. Now if there exists an elment $b\in \mathbb{F}_{2^{n}}$ such that $f^{\ast}_{1}(b)+f^{\ast}_{2}(b)+f^{\ast}_{3}(b)+f^{\ast}_{4}(b)=1$ (addition modulo 2), set $t_{b}:=f^{\ast}_{1}(b)+f^{\ast}_{2}(b)+f^{\ast}_{3}(b)+f^{\ast}_{4}(b)$, then $t_{b}\in \{1,3\}$ (recall here that these sums are calculated in $\mathbb{Z}$). We have  \begin{eqnarray*} W_{\sigma}(b) &\!=\!& 2^{k-1}((-1)^{f_{1}^{\ast}(b)}+(-1)^{f_{2}^{\ast}(b)}+(-1)^{f_{3}^{\ast}(b)}-(-1)^{f_{4}^{\ast}(b)})\\
             &\!=\!& 2^{k}(1-(f^{\ast}_{1}(b)+f^{\ast}_{2}(b)+f^{\ast}_{3}(b)-f^{\ast}_{4}(b)))\\
             &\!=\!& 2^{k}(2f^{\ast}_{4}(b)+1-t_{b}).
\end{eqnarray*}

\noindent Then \begin{eqnarray*}  W_{\sigma}(b)=
\begin{cases}
2^{k+1}f_{4}^{\ast}(b), ~~~~~~~~~~~~~if~t_{b}=1,\\
2^{k+1}(f_{4}^{\ast}(b)-1), ~~~~~if~t_{b}=3,
\end{cases}
\end{eqnarray*}
which means that $W_{\sigma}(b)\in \{0,\pm 2^{k+1}\}$. Then we have that if  $f^{\ast}_{1}+f^{\ast}_{2}+f^{\ast}_{3}+f^{\ast}_{4}=1$, then $\sigma$ is semi-bent. It needs to show that the converse also holds true. To the contrary, if there exists an element $a\in \mathbb{F}_{2^{n}}$ such that $f^{\ast}_{1}(a)+f^{\ast}_{2}(a)+f^{\ast}_{3}(a)+f^{\ast}_{4}(a)=0$, then by (\ref{0-pla}), one has $W_{\sigma}(a)=\pm 2^{k}$, a contradiction with the assumption that $\sigma$ is semi-bent !  Thus the assertion 2) holds true.

 We in fact have proved that for any $a \in \mathbb{F}_{2^{n}}$, \begin{eqnarray}\label{1-pla} W_{\sigma}(a)\in \{0, \pm 2^{k+1}\} \text{~if and only if~} f^{\ast}_{1}(a)+f^{\ast}_{2}(a)+f^{\ast}_{3}(a)+f^{\ast}_{4}(a)\equiv 1~({\rm mod}~2).
\end{eqnarray}
Now by Parseval's relation and (\ref{0-pla}), (\ref{1-pla}), one can obtain the assertion 3). We complete the proof. \end{proof}

Recall that the first-order derivative of an $(n,m)$-function $F$ is defined as $D_{a}F(x):=F(x)+F(x+a)$, and the second-order derivative of $F$ with respect to $(a,b)$ is defined as $D_{a}D_{b}F(x):=F(x)+F(x+a)+F(x+b)+F(x+a+b)$, where $a, b\in \mathbb{F}_{2^{n}}$.  Let $f(x)\in \mathcal{B}_{n}$ be any bent function, and $a_{1}, a_{2}, a_{3}$ be any three elements in $\mathbb{F}_{2^{n}}$. Let $f_{i}(x)=f(x)+{\rm Tr}^{n}_{1}(a_{i}x)$, then
\begin{eqnarray}\label{cal-dual}
f_{i} \text{~is bent with the dual function~} f_{i}^{\ast}(x)=f^{\ast}(x+a_{i}),
\end{eqnarray}
for $i=1,2,3$, see \cite{Cant03}. Mesnager has showed in \cite{Mes14} that $f^{\ast}_{1}+f^{\ast}_{2}+f^{\ast}_{3}+f^{\ast}_{4}$ $=D_{a_{1}+a_{2}}D_{a_{1}+a_{3}}f^{\ast}$, and  \begin{eqnarray*}\sigma&\!=\!&f_{1}f_{2}+f_{1}f_{3}+f_{2}f_{3} \\
             &\!=\!& f(x)\!+\!{\rm Tr}^{n}_{1}(a_{1}x){\rm Tr}^{n}_{1}(a_{2}x)\!+\!{\rm Tr}^{n}_{1}(a_{1}x){\rm Tr}^{n}_{1}(a_{3}x)\!+\!{\rm Tr}^{n}_{1}(a_{2}x){\rm Tr}^{n}_{1}(a_{3}x).
             \end{eqnarray*}
Thus, by Theorem \ref{startingpoint}, we have $\sigma$ is bent if and only if $D_{a_{1}+a_{2}}D_{a_{1}+a_{3}}f^{\ast}=0$;  $\sigma$ is semi-bent if and only if $D_{a_{1}+a_{2}}D_{a_{1}+a_{3}}f^{\ast}=1$.

\begin{prop}\label{mes-critical-point} Let $n=2k$ be a positive integer. Let $f(x)\in \mathcal{B}_{n}$ be any bent function, and $a_{1}, a_{2}, a_{3}$ be any three elements in $\mathbb{F}_{2^{n}}$. Let $f_{i}(x)=f(x)+{\rm Tr}^{n}_{1}(a_{i}x)$, $i=1,2,3$. Then $\sigma=f_{1}f_{2}+f_{1}f_{3}+f_{2}f_{3}$$=f(x)\!+\!{\rm Tr}^{n}_{1}(a_{1}x){\rm Tr}^{n}_{1}(a_{2}x)\!+\!{\rm Tr}^{n}_{1}(a_{1}x){\rm Tr}^{n}_{1}(a_{3}x)\!+\!{\rm Tr}^{n}_{1}(a_{2}x){\rm Tr}^{n}_{1}(a_{3}x)$ is bent if and only if $D_{a_{1}+a_{2}}D_{a_{1}+a_{3}}f^{\ast}=0$; $\sigma$ is semi-bent if and only if $D_{a_{1}+a_{2}}D_{a_{1}+a_{3}}f^{\ast}=1$.  If $\sigma$ is bent, then
\begin{eqnarray}\label{6-dual}
\sigma^{\ast}(x)=f^{\ast}(x+a_{1})f^{\ast}(x+a_{2})+f^{\ast}(x+a_{1})f^{\ast}(x+a_{3})+f^{\ast}(x+a_{2})f^{\ast}(x+a_{3}).
\end{eqnarray}
\end{prop}
\begin{proof} We need only to show the last assertion. However, this can be seen directly from that if $\sigma$ is bent, then by Theorem \ref{startingpoint}, $\sigma^{\ast}=f^{\ast}_{1}f^{\ast}_{2}+f^{\ast}_{1}f^{\ast}_{3}+f^{\ast}_{2}f^{\ast}_{3}$ and the fact (\ref{cal-dual}).
\end{proof}

With the notations as in the proposition above, let $a=a_{1}+a_{2}$, $b=a_{1}+a_{3}$. Then $\sigma$ is reduced to $f(x)+{\rm Tr}^{n}_{1}(ax){\rm Tr}^{n}_{1}(bx)+{\rm Tr}^{n}_{1}(a_{1}x)$. Let $h(x)=f(x)+{\rm Tr}^{n}_{1}(ax){\rm Tr}^{n}_{1}(bx)$. Then by Proposition \ref{mes-critical-point}, $h(x)$ is bent if and only if $D_{a}D_{b}f^{\ast}=0$; In this case,   $h^{\ast}(x)=f^{\ast}(x)f^{\ast}(x+a)+f^{\ast}(x)f^{\ast}(x+b)+f^{\ast}(x+a)f^{\ast}(x+b)$, see also \cite[Corollary 5]{Mes14}. And $h(x)$ is semi-bent if and only if $D_{a}D_{b}f^{\ast}=1$. In fact, we have proved the following corollary.

\begin{coro}\label{mes-starting-point} Let $n=2k$ be a positive integer. Let $f(x)\in \mathcal{B}_{n}$ be any bent function, and $a, b$ be any two elements in $\mathbb{F}_{2^{n}}$ with $a\neq b$. Then  $h(x)=f(x)+{\rm Tr}^{n}_{1}(ax){\rm Tr}^{n}_{1}(bx)$ is bent if and only if  $D_{a}D_{b}f^{\ast}=0$; $h(x)$ is semi-bent if and only if $D_{a}D_{b}f^{\ast}=1$. If $h(x)$ is bent, then $h^{\ast}(x)=f^{\ast}(x)f^{\ast}(x+a)+f^{\ast}(x)f^{\ast}(x+b)+f^{\ast}(x+a)f^{\ast}(x+b)$.
\end{coro}

Now, we want to see what will happen if there exist three pairwise distinct elements $a,b,c \in \mathbb{F}_{2^{n}}$ such that $D_{a}D_{b}f^{\ast}(x)=0$, $D_{a}D_{c}f^{\ast}(x)=0$, and $D_{b}D_{c}f^{\ast}(x)=0$. Let $f_{1}(x)=f(x)+{\rm Tr}^{n}_{1}(ax){\rm Tr}^{n}_{1}(bx)$, $f_{2}(x)=f(x)+{\rm Tr}^{n}_{1}(ax){\rm Tr}^{n}_{1}(cx)$, $f_{3}(x)=f(x)+{\rm Tr}^{n}_{1}(bx){\rm Tr}^{n}_{1}(cx)$. Then $f_{4}(x)=f_{1}(x)+f_{2}(x)+f_{3}(x)=f(x)+{\rm Tr}^{n}_{1}(ax){\rm Tr}^{n}_{1}(bx)+{\rm Tr}^{n}_{1}(ax){\rm Tr}^{n}_{1}(cx)+{\rm Tr}^{n}_{1}(ax){\rm Tr}^{n}_{1}(bx).$ By Proposition \ref{mes-critical-point}, $f_{4}(x)$ is bent if and only if $D_{a+b}D_{a+c}f^{\ast}(x)=0$. However, this is indeed the case by the assumptions and Lemma \ref{2-derivative-0} below. It means that $f_{i}$ is a bent function for $i=1,2,3,4$. Thus by Theorem \ref{startingpoint}, $\sigma(x)=f_{1}f_{2}+f_{1}f_{3}+f_{2}f_{3}=f(x)+{\rm Tr}^{n}_{1}(ax){\rm Tr}^{n}_{1}(bx){\rm Tr}^{n}_{1}(cx)$ is bent if and only if $f^{\ast}_{1}+ f^{\ast}_{2}+ f^{\ast}_{3}+ f^{\ast}_{4}=0$. We have to calculate the dual functions $f^{\ast}_{i}(x), i=1,2,3,4$. According to Corollary \ref{mes-starting-point}, we have
\begin{eqnarray*}\begin{cases}f_{1}^{\ast}(x)=f^{\ast}(x)f^{\ast}(x+a)+f^{\ast}(x)f^{\ast}(x+b)+f^{\ast}(x+a)f^{\ast}(x+b), \\ f_{2}^{\ast}(x)=f^{\ast}(x)f^{\ast}(x+a)+f^{\ast}(x)f^{\ast}(x+c)+f^{\ast}(x+a)f^{\ast}(x+c),\\
f_{3}^{\ast}(x)=f^{\ast}(x)f^{\ast}(x+b)+f^{\ast}(x)f^{\ast}(x+c)+f^{\ast}(x+b)f^{\ast}(x+c),
\end{cases} \end{eqnarray*} and hence
\begin{eqnarray*}f_{1}^{\ast}(x)+f_{2}^{\ast}(x)+f_{3}^{\ast}(x)=f^{\ast}(x+a)f^{\ast}(x+b)+f^{\ast}(x+a)f^{\ast}(x+c)+f^{\ast}(x+b)f^{\ast}(x+c).
 \end{eqnarray*} On the other hand, by  (\ref{6-dual}), we have $f^{\ast}_{4}(x)=f^{\ast}(x+a)f^{\ast}(x+b)+f^{\ast}(x+a)f^{\ast}(x+c)+f^{\ast}(x+b)f^{\ast}(x+c)$, that is, $f^{\ast}_{1}+ f^{\ast}_{2}+ f^{\ast}_{3}+ f^{\ast}_{4}=0$. Then $\sigma(x)=f(x)+{\rm Tr}^{n}_{1}(ax){\rm Tr}^{n}_{1}(bx){\rm Tr}^{n}_{1}(cx)$ is bent, and \begin{eqnarray*}\sigma^{\ast}(x)&=&f^{\ast}_{1}(x)f^{\ast}_{2}(x)+f^{\ast}_{1}(x)f^{\ast}_{3}(x)+f^{\ast}_{2}(x)f^{\ast}_{3}(x)\\
 &=&f^{\ast}(x)+g_{1}(x)g_{2}(x)g_{3}(x),
  \end{eqnarray*} where $g_{1}(x)=D_{a}f^{\ast}(x)$, $g_{2}(x)=D_{b}f^{\ast}(x)$, $g_{3}(x)=D_{c}f^{\ast}(x)$. By the arguments above and Corollary \ref{mes-starting-point},  we have the following results which can infer the main results of \cite{Xu17}.

\begin{coro}\label{cubic} Let $n=2k$ be a positive integer. Let $f(x)\in \mathcal{B}_{n}$ be any bent function, and $a, b, c$
be three pairwise distinct elements in $\mathbb{F}_{2^{n}}$.

1) If $D_{a}D_{b}f^{\ast}(x)=D_{a}D_{c}f^{\ast}(x)=D_{b}D_{c}f^{\ast}(x)=0$, then $\sigma(x)=f(x)+{\rm Tr}^{n}_{1}(ax){\rm Tr}^{n}_{1}(bx){\rm Tr}^{n}_{1}(cx)$ is a bent function with its dual $f^{\ast}(x)+g_{1}(x)g_{2}(x)g_{3}(x)$, where $g_{1}(x)=D_{a}f^{\ast}(x)$, $g_{2}(x)=D_{b}f^{\ast}(x)$, $g_{3}(x)=D_{c}f^{\ast}(x)$.

2) If $D_{a}D_{b}f^{\ast}(x)=0$, then $h(x)=f(x)+{\rm Tr}^{n}_{1}(ax){\rm Tr}^{n}_{1}(bx)$ is bent; if $D_{a}D_{b}f^{\ast}(x)=1$, then $h(x)$ is semi-bent; Otherwise, $h(x)$ is a function such that $\{|W_{h}(\nu)|~|~\nu\in \mathbb{F}_{2^{n}}\}=\{0,2^{k},2^{k+1}\}$.
\end{coro}
\begin{rmk}\label{reducepoly} With the same assumptions and notations as in the first assertion of corollary above and using Carlet-Mesnager's criterion, one can obtain the following interesting facts by selecting suitable bent functions $f_{1}, f_{2}, f_{3}$. Let $F(X_{1},X_{2},X_{3})$ be any reduced polynomials in $\mathbb{F}_{2}[X_{1},X_{2},X_{3}]$ (we send the readers to Section 2 concerning the definition of reduced polynomials), then $f(x)+F({\rm Tr}^{n}_{1}(ax),{\rm Tr}^{n}_{1}(bx), {\rm Tr}^{n}_{1}(cx))$ is bent with its dual $f^{\ast}(x)+F(g_{1}(x),g_{2}(x),g_{3}(x))$.
\end{rmk}

\noindent {\bf B. Property ($\mathbf{P}_{\tau}$) and equivalent conditions}\\

In this subsection, we introduce property ($\mathbf{P}_{\tau}$) concerning Boolean functions. Inspired by the observations made by Corollary \ref{cubic} and Remark \ref{reducepoly}, we want to consider more general cases. Explicitly, for a given Boolean function $g(x)\in \mathcal{B}_{n}$, we wonder to know what will happen if there exist $\tau$~($\tau \geq 2$) pairwise distinct elements $u_{i}$ such that $D_{u_{i}}D_{u_{j}}g(x)=0$, $\forall~1\leq i<j\leq \tau$. To this end, we introduce the property ($\mathbf{P}_{\tau}$). We will deduce new vectorial bent and plateaued functions starting from the observations on (bent) Boolean functions satisfying this property in next section, and we  believe that this property has its own value.

\begin{defn} Let $n,\tau $ be two positive integers. Let $g(x)\in \mathcal{B}_{n}$, and  $g$ is {\it said to satisfy property} $\mathbf{(P_{\tau})}$ if there exist $\tau$ pairwise distinct elements $u_{1},\ldots,u_{\tau}\in \mathbb{F}_{2^{n}}$ such that $D_{u_{i}}D_{u_{j}}g(x)=0$ for any $1\leq i<j\leq \tau$. In this case, the set $\{u_{1},\ldots,u_{\tau} \} \subseteq \mathbb{F}_{2^{n}} $ is called the defining set of $g(x)$ satisfying property $\mathbf{(P_{\tau})}$.
\end{defn}

 In the following, we give some observations on functions $g(x)$ satisfying  property $\mathbf{(P_{\tau})}$. We will not specify the subfix $\tau$ in case there is no danger of confusion. We need the following lemma.

\begin{lem}\label{2-derivative-0} For any $a,b,c\in \mathbb{F}_{2^{n}}$, if $D_{a}D_{b}g(x)=D_{a}D_{c}g(x)=0$ for all $x\in \mathbb{F}_{2^{n}}$, then $D_{a}D_{b+c}g(x)=0$ for all  $x\in \mathbb{F}_{2^{n}}$. Furthermore, if there exists $\{u_{1},\ldots, u_{\tau}\}\subseteq \mathbb{F}_{2^{n}}$ such that $D_{u_{i}}D_{u_{j}}g(x)=0$ for any $1\leq i<j\leq \tau$, then for any $a, b \in L(u_{1},\ldots,u_{\tau})$, we have $D_{a}D_{b}g(x)=0$, where $L(u_{1},\ldots,u_{\tau})$ is the subspace of $\mathbb{F}_{2^{n}}$ spanned by $\{u_{1},\ldots, u_{\tau}\}$ over $\mathbb{F}_{2}$.
\end{lem}
\begin{proof} By assumption, we have
\begin{eqnarray*}
\begin{cases}
g(x)\!+\!g(x\!+\!a)+g(x\!+\!b)+g(x\!+\!a\!+\!b)\!=\!0,\\
g(x)\!+\!g(x\!+\!a)+g(x\!+\!c)+g(x\!+\!a\!+\!c)\!=\!0,
\end{cases}
\end{eqnarray*}
for all $x\in \mathbb{F}_{2^{n}}$. Then $g(x+b)+g(x+c)+g(x+a+b)+g(x+a+c)=0$ for all $x\in \mathbb{F}_{2^{n}}$. We have by replacing $x+b$ by $x$ that $g(x)+g(x+b+c)+g(x+a)+g(x+b+a+c)=0$, i.e., $D_{a}D_{b+c}g(x)=0$ for all $x\in \mathbb{F}_{2^{n}}$. The last assertion follows from the first assertion and the fact that for any $a\in \mathbb{F}_{2^{n}}$, $x\in \mathbb{F}_{2^{n}}$,  $D_{a}D_{a}g(x)=0$.
\end{proof}

\begin{rmk} From a given Boolean function $g(x)$ satisfying property ($\mathbf{P}_{\tau}$), one can obtain a lot of other functions satisfying this property with the same defining set as $g(x)$. Indeed, let $g(x)\in \mathcal{B}_{n}$ be any Boolean function satisfying property ($\mathbf{P_{\tau}}$) with defining set $\{u_{1},\ldots,u_{\tau} \} \subseteq \mathbb{F}_{2^{n}} $. For any $b\in L(u_{1},\ldots,u_{\tau})$, set $h(x):=g(x)g(x+b)$, then $h(x)$ is also a Boolean function satisfying Property ($\mathbf{P_{\tau}}$) with the same defining set. To see this, it needs only to show $D_{u_{i}}D_{u_{j}}h(x)=0$ for any $1\leq i<j\leq \tau $. Note that for any $\varrho\in L(u_{1},\ldots,u_{\tau})$, by Lemma \ref{2-derivative-0}, we have $g(x+\varrho)g(x+\varrho+b)=g(x+\varrho)(g(x)+g(x+\varrho)+g(x+b))=g(x)g(x+\varrho)+g(x+\varrho)+g(x+\varrho)g(x+b)$. Then one has \begin{eqnarray*}\!D_{u_{i}}D_{u_{j}}h(x)&\!=\!&g(x)g(x\!+\!b)\!+\!g(x\!+\!u_{i})g(x\!+\!u_{i}\!+\!b)\!+\!g(x+u_{j})g(x\!+\!u_{j}\!+\!b)+g(x\!+\!u_{i}\!+\!u_{j})g(x\!+\!u_{i}\!+\!u_{j}\!+\!b)\!\\
&\!=\!&g(x)g(x\!+\!b)\!+\!g(x\!+\!u_{i})\!+\!g(x\!+\!u_{j})\!+\!g(x\!+\!u_{i}\!+\!u_{j})\\
&&\!+\!g(x)(g(x\!+\!u_{i})\!+\!g(x\!+\!u_{j})\!+\!g(x\!+\!u_{i}\!+\!u_{j}))\!+\!g(x\!+\!b)(g(x\!+\!u_{i})\!+\!g(x\!+\!u_{j})\!+\!g(x\!+\!u_{i}\!+\!u_{j}))\\
&\!=\!&g(x)g(x\!+\!b)\!+\!g(x)\!+\!g(x)\!+\!g(x\!+\!b)g(x)\!=\!0.
\end{eqnarray*}
\end{rmk}

The following observation is vital to our constructions of new vectorial bent functions. In fact, this observation establishes a link between property ($\mathbf{P}_{\tau}$) and the condition of Construction 7 in \cite{Tang17} which we will recall in the following section.

\begin{lem}\label{keylemma1} Let $g(x)\in \mathcal{B}_{n}$ be any Boolean function. The following two assertions are equivalent:

1) $g(x)$ satisfies property $\mathbf{(P_{\tau})}$ with the defining set $\{u_{1},\ldots,u_{\tau}\}\subseteq \mathbb{F}_{2^{n}}$.

2) there exist $u_{1},\ldots,u_{\tau}\in \mathbb{F}_{2^{n}}$, and $g_{1},\ldots, g_{\tau}\in \mathcal{B}_{n}$ such that $g(x+\sum\limits^{\tau}_{i=1}w_{i}u_{i})=g(x)+\sum\limits^{\tau}_{i=1}w_{i}g_{i}(x)$ for any $w=(w_{1},\ldots,w_{\tau})\in \mathbb{F}^{\tau}_{2}$.

Furthermore, if $g(x)$ satisfies property $\mathbf{(P_{\tau})}$ with the defining set $\{u_{1},\ldots,u_{\tau}\}\subseteq \mathbb{F}_{2^{n}}$, then  the $g_{i}(x)$ in 2) is exactly $D_{u_{i}}g(x), ~i=1,\ldots,\tau$.
\end{lem}
\begin{proof} $1)\Rightarrow 2)$: By assumption, there exist $u_{1},\ldots,u_{\tau} \in \mathbb{F}_{2^{n}}$ such that $D_{u_{i}}D_{u_{j}}g(x)=0$ for any $1\leq i<j\leq \tau$, and all $x\in \mathbb{F}_{2^{n}}$. Set $g_{i}(x):=D_{u_{i}}g(x), i=1,\ldots, \tau$. We will give our proof by induction on $s=wt(w)$. For $s=1$, we have $g(x+u_{i})=g(x)+g_{i}(x)$ by the definition of $g_{i}(x)$, for any $i=1,\ldots,\tau$. Consider the case of $s=2$: for  any $1\leq i<j\leq \tau$, one has $$g(x+u_{i}+u_{j})=g(x)+g(x+u_{i})+g(x+u_{j})=g(x)+D_{u_{i}}g(x)+D_{u_{j}}g(x)=g(x)+g_{i}(x)+g_{j}(x),$$ where the first identity is due to the assumption that $D_{u_{i}}D_{u_{j}}g(x)=0$. Now assume that the assertion holds for any $1\leq s\leq \tau-1$, that is,
$$g(x+u_{i_{1}}+\cdots+u_{i_{s}})=g(x)+g_{i_{1}}(x)+\cdots+g_{i_{s}}(x),  {\textit where~}  \{i_{1},\ldots,i_{s}\}\subseteq \{1,\ldots, \tau\}.$$
Then for any $w\in \mathbb{F}^{\tau}_{2}$ with $wt(w)=s+1$, we have
\begin{eqnarray*}g(x\!+\!u_{i_{1}}\!+\!\cdots+u_{i_{s}}\!+\!u_{i_{s+1}}) &\!=\!& g((x+u_{i_{1}})+u_{i_{2}}+\cdots+u_{i_{s}}+u_{i_{s+1}}), \\           &\!=\!&g(x+u_{i_{1}})+g_{i_{2}}(x+u_{i_{1}})+\cdots+g_{i_{s+1}}(x+u_{i_{1}}),\\
             &\!=\!&g(x)+g_{i_{1}}(x)+g_{i_{2}}(x)+\cdots +g_{i_{s+1}}(x),
             \end{eqnarray*}
where the second equality is from the induction on $s$, and the last equality is deduced by the definition of $g_{i}$, $i=1,2,\ldots,\tau$, and the induction on $s$ of the cases $s=1,~2$:  $g_{i_{t}}(x+u_{i_{1}})=g(x+u_{i_{1}})+g(x+u_{i_{1}}+u_{i_{t}})=g(x)+g_{i_{1}}(x)+g(x)+g_{i_{1}}(x)+g_{i_{t}}(x)=g_{i_{t}}(x).$

 $2)\Rightarrow 1)$: Let $\varepsilon^{1}=(1,0,\ldots,0),~\varepsilon^{2}=(0,1,\ldots,0),$  ~$\ldots,~\varepsilon^{\tau}=(0,0,\ldots,1)$ be the basis of $\mathbb{F}^{\tau}_{2}$. Let $w=\varepsilon^{i}$. Then by assumption we have $g_{i}(x)=D_{u_{i}}g(x)$, $i=1,\ldots,\tau$. For any $1\leq i<j\leq \tau$,  let $w=\varepsilon^{i}+\varepsilon^{j}$, we have $g(x+u_{i}+u_{j})=g(x)+g_{i}(x)+g_{j}(x)$, that is, $g(x+u_{i}+u_{j})=g(x)+D_{u_{i}}g(x)+D_{u_{j}}g(x)=g(x)+g(x+u_{i})+g(x+u_{j})$. Then $g(x)+g(x+u_{i})+g(x+u_{j})+g(x+u_{i}+u_{j})=0$, i.e., $D_{u_{i}}D_{u_{j}}g(x)=0$ for any $1\leq i<j\leq \tau$. We are done. \end{proof}

\section{Generic constructions of vectorial bent and plateaued functions}

In this section, we will construct new vectorial bent functions of the form \eqref{cons1} from known vectorial bent functions.  At first we give the following theorem.

\begin{thm}\label{cons-1} Let $n$ be an even positive integer and $m$ be a positive divisor of $n$. Let $G(x)$ be a vectorial bent $(n,m)$-function, and let $g(x)\in \mathcal{B}_{n}$. Then $H(x)=G(x)+g(x)$ is a vectorial bent (plateaued) function if and only if for any  $\lambda \in \mathbb{F}^{\ast}_{2^{m}}$ such that ${\rm Tr}^{m}_{1}(\lambda)=1$, $G_{\lambda}(x)+g(x)$ is a bent (plateaued) Boolean function.
\end{thm}

\begin{proof} For any $\lambda\in \mathbb{F}^{\ast}_{2^{m}}$, we have $H_{\lambda}(x)={\rm Tr}^{m}_{1}(\lambda H(x))={\rm Tr}^{m}_{1}(\lambda G(x))+{\rm Tr}^{m}_{1}(\lambda)g(x)$, and thus
\begin{eqnarray*}H_{\lambda}(x)=\begin{cases}G_{\lambda}(x), &{\rm~if~Tr}^{m}_{1}(\lambda)=0,\\ G_{\lambda}(x)+g(x), &{\rm~if~Tr}^{m}_{1}(\lambda)=1.
 \end{cases}
 \end{eqnarray*} Therefore, by definition $H(x)$ is a vectorial bent (plateaued) $(n,m)$-function if and only if for all $\lambda \in \mathbb{F}^{\ast}_{2^{m}}$ with ${\rm Tr}^{m}_{1}(\lambda)=1$, $G_{\lambda}(x)+g(x)$ is bent (plateaued), since $G(x)$ is vectorial bent.
\end{proof}

At a first glance, it would appear that finding such functions $G(x)$ and $g(x)$ satisfying the conditions of Theorem \ref{cons-1} might be quite difficult. However, our Corollary \ref{vecbent} below shows that, out of reckoning, there are quite a lot of such functions after we obtain Lemma \ref{keylemma1}, in which we establish a link between property $\mathbf{(P_{\tau})}$ and the condition of Construction 7 in  \cite{Tang17}. In what follows, let us recall the Construction 7 of \cite{Tang17}, in which the authors have a very nice observation on generating new bent functions from known ones.

Let $n=2k$, and  $u_{1},\ldots, u_{\tau} $ be distinct elements of $\mathbb{F}_{2^{n}}$, where $\tau$ is an integer with $1\leq \tau \leq k$. Let $g(x)\in \mathcal{B}_{n}$ be a bent function whose dual $g^{\ast}(x)$ satisfies that $g^{\ast}(x+\sum\limits^{\tau}_{i=1}w_{i}u_{i})=g^{\ast}(x)+\sum\limits^{\tau}_{i=1}w_{i}g_{i}(x)$ for any $x\in \mathbb{F}_{2^{n}}$ and for any $w=(w_{1},\ldots,w_{\tau})\in \mathbb{F}^{\tau}_{2}$, where $g_{i}(x)\in \mathcal{B}_{n}$ for any $1\leq i\leq \tau$. Let $F(X_{1},\ldots, X_{\tau})$ be any reduced polynomial in  $ \mathbb{F}_{2}[X_{1},\ldots,X_{\tau}]$. Then by \cite[Theorem 8]{Tang17}, $f(x):=g(x)+F({\rm Tr}^{n}_{1}(u_{1}x), {\rm Tr}^{n}_{1}(u_{2}x)), \ldots,{\rm Tr}^{n}_{1}(u_{\tau}x))$ is bent, with its dual $f^{\ast}(x)=g^{\ast}(x)+F(g_{1}(x),\ldots,g_{\tau}(x))$. In other words, using Lemma \ref{keylemma1}, the function $g(x)$ described in \cite[Construction 7]{Tang17} is a bent function such that its dual $g^{\ast}(x)$ satisfies property $\mathbf{(P_{\tau})}$ with the defining set $\{u_{1},\ldots,u_{\tau}\}$. In fact,  we have proved the following theorem.

\begin{thm}\label{tang} Let $n=2k$. Let $g(x)\in \mathcal{B}_{n}$ be a bent function such that its dual function $g^{\ast}(x)$ satisfies property  $\mathbf{(P_{\tau})}$ with the defining set  $\{u_{1},\ldots,u_{\tau}\}$. Let $F(X_{1},\ldots, X_{\tau})$ be any reduced polynomial in  $ \mathbb{F}_{2}[X_{1},\ldots,X_{\tau}]$.  Then the Boolean function $g(x)+F({\rm Tr}^{n}_{1}(u_{1}x), {\rm Tr}^{n}_{1}(u_{2}x)), \ldots,{\rm Tr}^{n}_{1}(u_{\tau}x))$ is bent, with its dual $g^{\ast}(x)+F(D_{u_{1}}g^{\ast}(x),\ldots,D_{u_{\tau}}g^{\ast}(x))$.
\end{thm}

\begin{rmk} We have to point out that though the authors in \cite{Tang17} give a nice secondary construction of bent functions  from bent functions $g(x)$ whose dual $g^{\ast}(x)$ satisfies the condition of the Construction 7 in \cite{Tang17}, they do not give any additional insights on this condition. We believe our property $\mathbf{(P_{\tau})}$  gives a quick and effective way to judge whether a given bent function satisfies this condition.
\end{rmk}

By Theorem \ref{cons-1}, and Theorem \ref{tang}, we can give a new secondary construction of vectorial bent functions.

\begin{coro}\label{vecbent}  Let $n=2k$ be an even positive integer, and $m$ be a positive divisor of $n$. Let $u_{1},\ldots,u_{\tau}\in\mathbb{F}_{2^{n}}$ be distinct, where $1\leq \tau \leq k$. Let $F(X_{1},\ldots, X_{\tau})$ be a reduced polynomial in $ \mathbb{F}_{2}[X_{1},\ldots,X_{\tau}]$. Assume that $G(x)$ is a vectorial bent $(n,m)$-function such that for any $\lambda\in \mathbb{F}_{2^{k}}$ with ${\rm Tr}^{k}_{1}(\lambda)=1$, the function $G^{\ast}_{\lambda}(x)$ satisfies property $(\mathbf{P}_{\tau})$ with the defining set $\{u_{1},\ldots,u_{\tau}\}$, then $H(x):=G(x)+F({\rm Tr}^{n}_{1}(u_{1}x), {\rm Tr}^{n}_{1}(u_{2}x)), \ldots,{\rm Tr}^{n}_{1}(u_{\tau}x))$ is a vectorial bent $(n,m)$-function.
\end{coro}
\begin{proof}  By Theorem \ref{cons-1}, it need only to show that for each $\lambda\in \mathbb{F}^{\ast}_{2^{m}}$ with ${\rm Tr}^{m}_{1}(\lambda)=1$,
$$G_{\lambda}(x)+F({\rm Tr}^{n}_{1}(u_{1}x), {\rm Tr}^{n}_{1}(u_{2}x)), \ldots,{\rm Tr}^{n}_{1}(u_{\tau}x))$$
is bent. Since $G^{\ast}_{\lambda}(x)$ satisfies property $(\mathbf{P}_{\tau})$ with the defining set $\{u_{1},\ldots,u_{\tau}\}$, say $D_{u_{i}}D_{u_{j}}G_{\lambda}^{\ast}(x)=0$ for any $1\leq i<j\leq \tau$. By Theorem \ref{tang}, $G_{\lambda}(x)+F({\rm Tr}^{n}_{1}(u_{1}x), {\rm Tr}^{n}_{1}(u_{2}x)), \ldots,{\rm Tr}^{n}_{1}(u_{\tau}x))$  is bent.  \end{proof}

Thanks to Corollary \ref{vecbent}, we can give a secondary construction of vectorial plateaued functions.

\begin{coro}\label{vecplateaued} Assuming conditions of Corollary \ref{vecbent}. Let $t$ be a positive integer. Let $F_{i}(X_{1},\ldots, X_{\tau})$, $i=1,
\ldots,t$, be any reduced polynomials in $ \mathbb{F}_{2}[X_{1},\ldots,X_{\tau}]$. Denote $F_{i}({\rm Tr}^{n}_{1}(u_{1}x), \ldots,{\rm Tr}^{n}_{1}(u_{\tau}x))$ by $f_{i}(x)$ for each $i=1,\ldots,t$. Then $\widehat{H}(x)=(G(x),f_{1}(x),\ldots,f_{t}(x))$ is a vectorial plateaued $(n,m+t)$-function if and only if the $(n,t)$-function $(f_{1}(x),\ldots,f_{t}(x))$ is vectorial plateaued.
\end{coro}
\begin{proof} For $(\lambda,v)\in \mathbb{F}_{2^{m}}^*\times \mathbb{F}^{t}_{2}$,  according to Theorem \ref{tang} and Corollary \ref{vecbent}, $\langle(\lambda,v), \widehat{H}\rangle$=$G_{\lambda}+ \langle v, (f_{1},\ldots,f_{t})\rangle$ is bent, since $\langle v, (F_{1},\ldots,F_{t})\rangle$ is also a reduced polynomial in $ \mathbb{F}_{2}[X_{1},\ldots,X_{\tau}]$. Then $\widehat{H}$ is vectorial plateaued if and only if all the components functions   $\langle(0,v), \widehat{H}\rangle=\langle v, (f_{1},\ldots,f_{t})\rangle$ is plateaued. It means that  the $(n,t)$-function $(f_{1}(x),\ldots,f_{t}(x))$ is vectorial plateaued. This completes the proof.   \end{proof}

\section{New infinite families of vectorial bent and plateaued functions}

In this section, using the results from the previous section, we will obtain (at least) three classes of new primary constructions of vectorial bent and vectorial plateaued functions. Amongst those vectorial plateaued functions, there are two classes of functions having the maximal number of bent components. \\

\noindent $A$. \textbf{ New infinite families of vectorial bent functions via Kasami function} \\

Let $n=2k$ be an even positive integer throughout this subsection. Let $G(x)=x^{2^{k}+1}$. It is well known that $G$ is a vectorial bent $(n,k)$-function. The dual of its component $G_{\lambda}(x)={\rm Tr}^{k}_{1}(\lambda G(x))$, for some $\lambda \in \mathbb{F}^{\ast}_{2^{k}}$, is $G^{\ast}_{\lambda}(x)={\rm Tr}^{k}_{1}(\lambda^{-1}x^{2^{k}+1})+1$ (see \cite{Mes14}).

Now, in order to apply Corollary \ref{vecbent}, one has to find a set $\{\rm u_{1}, \ldots, u_{\tau} \}\subseteq \mathbb{F}_{2^{n}}$ such that for all $\lambda \in \mathbb{F}^{\ast}_{2^{k}}$ with ${\rm Tr}^{k}_{1}(\lambda)=1$, $D_{u_{i}}D_{u_{j}}G^{\ast}_{\lambda}(x)={\rm Tr}^{n}_{1}(\lambda^{-1}u_{i}\overline{u}_{j})=0$ for any $1\leq i<j\leq \tau$, where $\overline{u}_{j}:=u^{2^{k}}_{j}$. Note that for any  $\lambda \in \mathbb{F}^{\ast}_{2^{k}}$ with ${\rm Tr}^{k}_{1}(\lambda)=1$, the element $\lambda^{-1}$ can be represented by $v+\overline{v}$ for a unique set $\{v,\overline{v}~|~v\in U\}$, here $U=\{ x\in \mathbb{F}_{2^{n}}~|~x\overline{x}=1\}$. Then $D_{u_{i}}D_{u_{j}}G^{\ast}_{\lambda}(x)={\rm Tr}^{n}_{1}((v+\overline{v})u_{i}\overline{u}_{j})={\rm Tr}^{n}_{1}(v(u_{i}\overline{u}_{j}+\overline{u}_{i}u_{j}))$. Hence $D_{u_{i}}D_{u_{j}}G^{\ast}_{\lambda}(x)=0$ for all $\lambda \in \mathbb{F}^{\ast}_{2^{k}}$ with ${\rm Tr}^{k}_{1}(\lambda)=1$ if and only if ${\rm Tr}^{n}_{1}(v(u_{i}\overline{u}_{j}+\overline{u}_{i}u_{j})=0$ for all $v\in U$. It is easily seen that if $u_{i}\overline{u}_{j}+\overline{u}_{i}u_{j}=0$, i.e., $u_{i}\overline{u}_{j}\in \mathbb{F}_{2^{k}}$, then the conditions of Corollary \ref{vecbent} is automatically satisfied. In particular, let $\{\varrho_{1},\ldots,\varrho_{k}\}$ be a basis of $\mathbb{F}_{2^{k}}$ over $\mathbb{F}_{2}$, and $v\neq 1$ be an element of $U$, set $u_{i}:=\varrho_{i}v,$ $i=1,\ldots, k$,  then we have $u_{i}\overline{u}_{j}\in \mathbb{F}_{2^{k}}$ for any $1\leq i<j\leq k$.

\begin{thm}\label{a} Let $n=2k$ and $\tau$ be positive integers with $1\leq \tau \leq k$. Let $u_{1}, \ldots, u_{k}$ be any $k$ pairwise distinct elements in $\mathbb{F}_{2^{n}}$ such that $u_{i}u^{2^{k}}_{j}\in \mathbb{F}^{\ast}_{2^{k}}$ for any $1\leq i<j\leq k$. Let $F(X_{1},X_{2},\ldots, X_{\tau})$ be any reduced polynomial in $\mathbb{F}_{2}[X_{1},X_{2},\ldots, X_{\tau}]$ with algebraic degree $d$, where $d$ is a nonnegative integer.  Then $H(x)=x^{2^{k}+1}+F({\rm Tr}^{n}_{1}(u_{i_{1}}x), {\rm Tr}^{n}_{1}(u_{i_{2}}x), \ldots,{\rm Tr}^{n}_{1}(u_{i_{\tau}}x))$ is a vectorial bent function, where $\{i_{1},\ldots, i_{\tau} \} \subseteq \{1,\ldots,k\}$. Furthermore, if $u_{i_{1}},\ldots, u_{i_{\tau}}$ are linearly
independent over $\mathbb{F}_{2}$ and $d\geq 2$, then the algebraic degree of $H(x)$ is equal to $d$.
\end{thm}
\begin{proof} It need only to show the last assertion. By assumption $u_{i_{1}},\ldots, u_{i_{\tau}}$ are linearly independent over $\mathbb{F}_{2}$, then according to Lemma 2 of \cite{Tang17}, the algebraic degree of $F({\rm Tr}^{n}_{1}(u_{i_{1}}x), {\rm Tr}^{n}_{1}(u_{i_{2}}x),\ldots, {\rm Tr}^{n}_{1}(u_{i_{\tau}}x))$ is equal to $d$.  \end{proof}

\begin{coro}\label{a-vp} Conditions are the same with Theorem \ref{a}. Let $F_{i}(X_{1},\ldots, X_{k})$, $i=1,
\ldots,t$, be any reduced polynomials in $ \mathbb{F}_{2}[X_{1},\ldots,X_{k}]$, for some positive integer $t$. Set $f_{i}(x):=F_{i}({\rm Tr}^{n}_{1}(u_{1}x)), \ldots,{\rm Tr}^{n}_{1}(u_{k}x))$ for each $i=1,\ldots,t$. Then $\widehat{H}(x)=(x^{2^{k}+1},f_{1}(x),\ldots,f_{t}(x))$ is a vectorial plateaued $(n,k+t)$-function if and only if the $(n,t)$-function $(f_{1}(x),\ldots,f_{t}(x))$ is vectorial plateaued.
\end{coro}
\begin{proof} This can be seen directly from Corollary \ref{vecplateaued} and Theorem \ref{a}.
\end{proof}

 It is important and interesting to estimate the number of the bent components of $\widehat{H}(x)$. Note that for any $(u,v)\in \mathbb{F}^{\ast}_{2^{k}}\times \mathbb{F}^{m}_{2}$, $\langle(u,v), \widehat{H}(x)\rangle={\rm Tr}^{k}_{1}(ux^{2^{k}+1})+\langle v,(f_{1}(x),\ldots,f_{t}(x)) \rangle$. Therefore, by the fact that $\langle v,(F_{1},\ldots,F_{t}) \rangle$ is also a polynomial over $\mathbb{F}_{2}$ with the variables $X_{1},\ldots, X_{k}$, we have by Theorem \ref{a}, $\langle(u,v), \widehat{H}\rangle$ is bent for any $u\neq 0$. It means that $\widehat{H}(x)$ has at least $2^{t+k}-2^{t}$ bent components. It is not hard to prove that (or see \cite[Theorem 3.2]{ZP18}), the maximal number of bent components for a $(2k,k+t)$-function is $2^{t+k}-2^{t}$. Therefore, $\widehat{H}(x)$ has $2^{t+k}-2^{t}$ bent components, and $\langle(u,v), \widehat{H}(x)\rangle$ is bent if and only if $u\neq 0$. We in fact have proved the following corollary.

 \begin{coro}\label{a-vm} With the same notations in Corollary \ref{a-vp}. For any $(u,v)\in \mathbb{F}^{\ast}_{2^{k}}\times \mathbb{F}^{t}_{2}$,  $\langle(u,v), \widehat{H}(x)\rangle$ is bent if and only if $u\neq 0$. In particular, $\widehat{H}(x)$ is an $(n,t+k)$-function with the maximal number of bent components, and for any $v\in  \mathbb{F}^{t}_{2}$, $\langle v,(f_{1}(x),\ldots,f_{t}(x)) \rangle$ is not a bent function.
 \end{coro}
\begin{rmk} With the same notations in Corollary \ref{a-vp}. Let $\widehat{H}(x)=(x^{2^{k}+1},{\rm Tr}^{n}_{1}(u_{1}x),{\rm Tr}^{n}_{1}(u_{1}x){\rm Tr}^{n}_{1}(u_{2}x),$ $\ldots, \prod\limits^{k}_{i=1}{\rm Tr}^{n}_{1}(u_{i}x))$. If $u_{1},\ldots, u_{k}$ are linearly independent over $\mathbb{F}_{2}$, then $\widehat{H}(x)$ is an $(n,n)$-function of algebraic degree $k$, and has the maximal number of bent components in the sense of \cite[Theorem 3.2]{ZP18}, see also \cite[Theorem 2]{Pott18}. It is interesting to investigate its cryptographic properties such as APN-ness etc. This will be the topic of our future work.
\end{rmk}

\noindent $B$. \textbf{ New infinite families of vectorial bent functions from Niho exponents}\\

Throughout this subsection, $n=2k$ is an even integer, and $\tau$ is a positive integer such that $1\leq \tau \leq k$. For any $a$ in $\mathbb{F}_{2^{n}}$, denote $a^{2^{k}}$ by $\overline{a}$. Consider now the $(n,n)$-function
$$G(x)=\sum\limits^{2^{r}-1}_{i=1}x^{(i2^{k-r}+1)(2^{k}-1)+1}$$ with $1<r<k$ and ${\rm gcd}(r,k)=1$, then by \cite[Theorem 2]{Li13}, for any $a\in \mathbb{F}_{2^{n}}$, $G_{a}(x)={\rm Tr}^{n}_{1}(aG(x))$ is bent if $a+\overline{a}\neq 0$. In the following, we first show that $G(x)$ is actually a vectorial bent $(n,k)$-function, and then use it to generate new vectorial bent $(n,k)$-functions of the form \eqref{cons1}.

\begin{prop} Let $n=2k$, $r$ be positive integers such that ${\rm gcd}(r,k)=1$. Then the $(n,k)$-function $G(x)=\sum\limits^{2^{r}-1}_{i=1}x^{(i2^{k-r}+1)(2^{k}-1)+1}$ is vectorial bent.
\end{prop}
\begin{proof} By \cite[Theorem~2]{Li13}, ${\rm Tr}^{n}_{1}(aG(x))$ is bent if $a+\overline{a}\neq 0$, that is,  $a\not\in \mathbb{F}_{2^{k}}$. Then according to \cite[Proposition 3]{Pott18}, ${\rm Tr}^{n}_{k}(aG(x))$ is a vectorial bent $(n,k)$-function for any $a\not\in \mathbb{F}_{2^{k}}$. Thus the assertion will become true if we can show that $G(x)\in \mathbb{F}_{2^{k}}$ for all $x\in \mathbb{F}_{2^{n}}$. Since if this is the case, let $a\in \mathbb{F}_{2^{n}}$ such that $a+\overline{a}=1$, then ${\rm Tr}^{n}_{k}(aG(x))=G(x){\rm Tr}^{n}_{k}(a)=G(x)$ is a vectorial bent $(n,k)$-function.

Indeed, let $d_{i}=(2^{k}-1)s_{i}+1$ with $s_{i}=i2^{k-r}+1$ for $i=1,\ldots, 2^{r}-1,$ then for any $1\leq i<j\leq 2^{r}-1$ such that $i+j=2^{r}$, it holds $d_{j}\equiv d_{i}\cdot2^{k}~({\rm mod}~2^{n}-1),$ and hence
\begin{eqnarray*}G(x)&=&x^{(2^{k}+1)2^{k-1}} +\sum\limits^{2^{r}-1}_{i=1,i\neq 2^{r-1}}x^{(i2^{k-r}+1)(2^{k}-1)+1}\\
&=&(x^{(2^{k}+1)2^{k-1}} +\sum\limits^{2^{r-1}-1}_{i=1}(x^{d_{i}}+ x^{d_{i}\cdot 2^{k}}))~({\rm mod}~x^{2^{n}}+x).
 \end{eqnarray*} Now it is easy to see that $G$ is an $(n,k)$-function, since $x^{(2^{k}+1)2^{k-1}}, ~x^{d_{i}}+ x^{d_{i}\cdot 2^{k}}={\rm Tr}^{n}_{k}(x^{d_{i}}) \in \mathbb{F}_{2^{k}}$ for any $x\in \mathbb{F}_{2^{n}}$ and $1 \leq i\leq 2^{r-1}-1$.
\end{proof}

Considering the vectorial bent $(n,k)$-function $G(x)$ described above, for any $\lambda \in \mathbb{F}^{\ast}_{2^{k}}$, we have $G_{\lambda}(x)={\rm Tr}^{k}_{1}(\lambda G(x))={\rm Tr}^{k}_{1}(\lambda x^{(2^{k}+1)2^{k-1}}) + {\rm Tr}^{n}_{1}(\lambda \sum\limits^{2^{r-1}-1}_{i=1}x^{(i2^{k-r}+1)(2^{k}-1)+1})$. In order to construct new vectorial bent functions of the form \eqref{cons1}, one has to calculate the dual $G^{\ast}_{\lambda}(x)$ for each $\lambda \in \mathbb{F}^{\ast}_{2^{k}}$.

 Let $\lambda=1$, then $G_{\lambda}(x)$ $={\rm Tr}^{k}_{1}(x^{(2^{k}+1)2^{k-1}})$ +
${\rm Tr}^{n}_{1}( \sum\limits^{2^{r-1}-1}_{i=1}x^{(i2^{k-r}+1)(2^{k}-1)+1})$ which is exactly  the Leander-Kholosha's class of bent functions~(see \cite{Leander06}). Take any $u\in \mathbb{F}_{2^{n}}$ with $u+\overline{u}=1$. Then it has been shown  the dual function $G^{\ast}_{1}(x)$ of $G_{1}(x)$ is given by $$G^{\ast}_{1}(x)={\rm Tr}^{k}_{1}((u(1+x+x^{2^{k}})+u^{2^{n-r}}+x^{2^{k}})(1+x+x^{2^{k}})^{1/(2^{r}-1)}),$$
where $1/(2^{r}-1)$ is interpreted modulo $2^{k}-1$, say it is a positive integer $s$ such that $(2^{r}-1)\cdot s\equiv 1~({\rm mod}~2^k-1)$ (see \cite[Theorem~ 1]{Buda12}). Let $t=2^{r-1}-1$, $d_{t}=(2^{k}-1)(t2^{k-r}+1)+1$. In \cite[Proposition 3]{Li13}, the authors have shown that ${\rm gcd}(d_{t},2^{n}-1)=1$, and for each $\lambda \in \mathbb{F}^{\ast}_{2^{k}}$,  there exists a unique element $\delta\in \mathbb{F}_{2^{n}}$ such that $\lambda=\delta^{d_{t}}$, and $G_{\lambda}(x)=G_{1}(\delta x).$ Here, one can see that $\delta\in \mathbb{F}_{2^{k}}$.

Now we are in position to give the dual $G^{\ast}_{\lambda}(x)$ of $G_{\lambda}(x)$ for each $\lambda \in \mathbb{F}^{\ast}_{2^{k}}$.

\begin{prop}\label{Li} Let $G(x)$ be the vectorial bent $(n,k)$-function described above. Then $G^{\ast}_{\lambda}(x)=G^{\ast}_{1}(\delta^{-1}x)$ for each $\lambda \in \mathbb{F}^{\ast}_{2^{k}}$, where $\delta$ is the unique element in $\mathbb{F}_{2^{n}}$ such that $\lambda=\delta^{d_{t}}$, $d_{t}=2^{r-1}-1$. In particular, for any $a,b\in \mathbb{F}^{\ast}_{2^{k}}$, $D_{a}D_{b}G^{\ast}_{\lambda}(x)=0.$
\end{prop}
\begin{proof} We begin our proof from two bent functions $g,h\in \mathcal{B}_{n}$ satisfying that $g(x)=h(\delta x)$ for some $\delta\in \mathbb{F}^{\ast}_{2^{n}}$. We will obtain that $h^{\ast}(x)=g^{\ast}(\delta x)$, and then the first assertion holds true when we take $g(x)=G_{\lambda}(x), h(x)=G_{1}(x).$ Let $a\in \mathbb{F}_{2^{n}}$, we have
\begin{eqnarray*}W_{h}(a) &\!=\!& \sum\limits_{y\in \mathbb{F}_{2^{n}}}(-1)^{h(y)+{\rm Tr}^{n}_{1}(ay)}\\
             &\!=\!&  \sum\limits_{x\in \mathbb{F}_{2^{n}}}(-1)^{h(\delta x)+{\rm Tr}^{n}_{1}(a\delta x)}\\
             &\!=\!&\sum\limits_{x\in \mathbb{F}_{2^{n}}}(-1)^{g(x)+{\rm Tr}^{n}_{1}(a\delta x)}\\
             &\!=\!& W_{g}(a\delta).
             \end{eqnarray*}
It follows that $2^{k}(-1)^{h^{\ast}(a)}=2^{k}(-1)^{g^{\ast}(\delta a)}$ for any $a\in \mathbb{F}_{2^{n}}$, and hence $h^{\ast}(a)=g^{\ast}(\delta a)$. Note that by the proof of Theorem 11 in \cite{Mes14}, one has $D_{a}D_{b}G^{\ast}_{1}(x)=0$ for any $a,b\in \mathbb{F}^{\ast}_{2^{k}}$. Then  $D_{a}D_{b}G^{\ast}_{\lambda}(x)=D_{a}D_{b}G^{\ast}_{1}(\delta^{-1}x)=D_{\delta^{-1}a}D_{\delta^{-1}b}G^{\ast}_{1}(y)=0$ with $y=\delta^{-1}x$, since $\delta \in \mathbb{F}^{\ast}_{2^{k}}$ by the fact $\lambda \in \mathbb{F}^{\ast}_{2^{k}}$, $\lambda=\delta^{d_{t}}$, gcd$(d_{t}, 2^n-1)=1$.   \end{proof}

\begin{thm}\label{b} Let $\{u_{1}, \ldots, u_{k}\}$ be a basis of~~$\mathbb{F}_{2^{k}}$ over $\mathbb{F}_{2}$, and $G(x)=\sum\limits^{2^{r}-1}_{i=1}x^{(i2^{k-r}+1)(2^{k}-1)+1}$ with $r>1$, $gcd(r,k)=1$. Let $F(X_{1},X_{2},\ldots, X_{\tau})$ be any reduced polynomial in $\mathbb{F}_{2}[X_{1},X_{2},\ldots, X_{\tau}]$ with algebraic degree $d$.  Then $H(x)=G(x)+F({\rm Tr}^{n}_{1}(u_{i_{1}}x), {\rm Tr}^{n}_{1}(u_{i_{2}}x), \ldots,{\rm Tr}^{n}_{1}(u_{i_{\tau}}x))$ is a vectorial bent $(n,k)$-function, where $\{i_{1},\ldots, i_{\tau} \} \subseteq \{1,\ldots,k\}$. Furthermore, if $d=k$, and the algebraic degree of $G$ is not equal to $k$, then $H(x)$ has algebraic degree $k$.
\end{thm}
\begin{proof} To show the first assertion, it needs only to show that $D_{u_{i}}D_{u_{j}}G^{\ast}_{\lambda}(x)=0$ for all $\lambda\in \mathbb{F}^{\ast}_{2^{k}}$ satisfying ${\rm Tr}^{k}_{1}(\lambda)=1$ and any $1\leq i<j\leq k$. However, this can be seen from Proposition \ref{Li}. Noting that $u_{i_{1}},\ldots, u_{i_{\tau}}$ are linearly independent over $\mathbb{F}_{2}$, we have that the algebraic degree of the univariate function $F({\rm Tr}^{n}_{1}(u_{i_{1}}x), {\rm Tr}^{n}_{1}(u_{i_{2}}x),\ldots, {\rm Tr}^{n}_{1}(u_{i_{\tau}}x))$ is equal to $d$ by Lemma 2 of \cite{Tang17}. However, for any vectorial bent function, its algebraic degree is at most $k$, hence the last assertion follows from the fact that the algebraic degree $d(H)$ of $H$ is equal to ${\rm max}(d(G), d)=k$. \end{proof}

\begin{coro}\label{b-vp} With the same conditions of Theorem \ref{b}. Let $t$ be a positive integer. Let $F_{i}(X_{1}, \ldots, X_{k})$, $i=1,
\ldots, t$, be any reduced polynomials in $\mathbb{F}_{2}[X_{1}, \ldots, X_{k}]$. Set $f_{i}(x):=F_{i}({\rm Tr}^{n}_{1}(u_{1}x)), \ldots,{\rm Tr}^{n}_{1}(u_{k}x))$ for each $i=1,\ldots,t$. Then $\widehat{H}(x)=(\sum\limits^{2^{r}-1}_{i=1}x^{(i2^{k-r}+1)(2^{k}-1)+1},f_{1}(x),\ldots,f_{t}(x))$ is a vectorial plateaued $(n,k+t)$-function if and only if the $(n,t)$-function $(f_{1}(x),\ldots,f_{t}(x))$ is vectorial plateaued. In particular, if $k>2$ and $f_{i}$  is a quadratic function for each $i=1,\ldots,t$, then $\widehat{H}(x)$ is a non-quadratic vectorial plateaued function.
\end{coro}
\begin{proof} The first assertion can be seen from Corollary \ref{vecplateaued} and Theorem \ref{b}. We need only to show the last assertion. It is well known that quadratic vectorial functions are plateaued, for instance see \cite{Carlet15}.  Since $r>1$ and $k>2$,  the algebraic degree of $H$ is greater than 2. Now, one can conclude that $\widehat{H}(x)$ is a non-quadratic vectorial plateaued function.  \end{proof}

\begin{rmk} With similar arguments as in Corollary \ref{a-vm}, one can prove that the function  $\widehat{H}(x)$ described above is a vectorial function with maximal number of bent components.
\end{rmk}

\noindent $C$. \textbf{ New infinite families of vectorial bent functions from Gold-Like monomial functions}\\

Throughout this subsection, $n=4k$ is a positive integer with $k\geq 2$. Mesnager \cite{Mes14} pointed out that the monomial function ${\rm Tr}^{n}_{1}(\lambda x^{2^{k}+1})$ is self-dual bent for any $\lambda\in \mathbb{F}_{2^{n}}$ satisfying $\lambda+\lambda^{2^{3k}}=1$. Inspired by this work, we consider in this subsection the $(n,n)$-function $x \mapsto x^{2^{k}+1}$. Note that in \cite[Theorem 3]{Xu18} the authors have shown  that for $a\in \mathbb{F}^{\ast}_{2^{n}}$, ${\rm Tr}^{n}_{k}(a x^{2^{k}+1})$ is a vectorial bent $(n,k)$-function if and only if $a\not\in \langle \varrho^{2^{k}+1} \rangle $, where $\varrho$ is a primitive element of $\mathbb{F}_{2^{n}}$, $\langle a \rangle$ is the cyclic subgroup of $\mathbb{F}^{\ast}_{2^{n}}$ generated by $a$. Denote $U=\{x\in \mathbb{F}_{2^{2k}}~|~x^{2^{k}+1}=1\}$. Let $\omega=\varrho^{(2^{k}-1)(2^{2k}+1)}$. Then it can be seen that $\omega\in U\backslash \{1\}\subseteq \mathbb{F}_{2^{2k}}$, $\omega\not\in \langle \varrho^{2^{k}+1} \rangle$, and $\omega+\omega^{2^{k}}\neq 0$. Let $G(x)={\rm Tr}^{4k}_{k}(\omega x^{2^{k}+1})$. Thus, $G$ is a vectorial bent $(n,k)$-function by Theorem 3 of \cite{Xu18}.

Now, in order to construct vectorial bent functions of the form \eqref{cons1}, firstly one has to determine the dual $G^{\ast}_{\lambda}(x)$ of $G_{\lambda}(x)={\rm Tr}^{k}_{1}(\lambda {\rm Tr}^{4k}_{k}(\omega x^{2^{k}+1}))$ $={\rm Tr}^{n}_{1}(\lambda \omega x^{2^{k}+1})$ for all $\lambda \in \mathbb{F}^{\ast}_{2^{k}}$ such that ${\rm Tr}^{k}_{1}(\lambda)=1$.  We need the following lemma which gives the dual $G^{\ast}_{\lambda}(x)$ for all $\lambda \in \mathbb{F}^{\ast}_{2^{k}}$.

\begin{lem} With the same notations above. Let $\lambda_{0}=(w+w^{2^{k}})^{-1}$. Then $\lambda_{0}\in \mathbb{F}^{\ast}_{2^{k}}$, ${\rm Tr}^{k}_{1}(\lambda_{0})=1$, and $G_{\lambda_{0}}(x)$ is a self-dual bent function.  For any $\lambda \in \mathbb{F}^{\ast}_{2^{k}}$, let $\delta$ be the unique element in $\mathbb{F}^{\ast}_{2^{k}}$ such that $\delta^{2^{k}+1}=\lambda\lambda^{-1}_{0}$. Then $G^{\ast}_{\lambda}(x)=G_{\lambda_{0}}(\delta^{-1}x)$. In particular, for any $a, b\in \mathbb{F}^{\ast}_{2^{2k}}$ satisfying $ab^{2^{k}}\in \mathbb{F}_{2^{k}}$, $D_{a}D_{b}G^{\ast}_{\lambda}(x)=0$ for all $\lambda \in \mathbb{F}^{\ast}_{2^{k}}$ such that ${\rm Tr}^{k}_{1}(\lambda)=1$.
\end{lem}
\begin{proof}  Since $\omega\in U\backslash \{1\}\subseteq \mathbb{F}_{2^{2k}}\backslash \mathbb{F}_{2^{k}} $, we have $w+w^{2^{k}}\neq 0$, $\lambda_{0}=(w+w^{2^{k}})^{-1}\in \mathbb{F}^{\ast}_{2^{k}}$, and ${\rm Tr}^{k}_{1}(\lambda_{0})=1$. Note that for $\lambda\in \mathbb{F}_{2^{n}}$, $\lambda+\lambda^{2^{3k}}=1$ is equivalent to $\lambda^{2^{k}}+\lambda=1$. For those $\lambda$, Mesnager has showed that ${\rm Tr}^{n}_{1}(\lambda x^{2^{k}+1})$ is a self-dual bent function \cite[Lemma 23]{Mes14}. Now, $\lambda_{0}\omega+(\lambda_{0}\omega)^{2^{k}}=\lambda_{0}(\omega+\omega^{2^{k}})=1$. Therefore, $G_{\lambda_{0}}(x)={\rm Tr}^{n}_{1}(\lambda_{0} \omega x^{2^{k}+1})$ is a self-dual bent function. Note that {\rm gcd}$(2^{k}+1, 2^{k}-1)$=1, for a given $\lambda \in \mathbb{F}^{\ast}_{2^{k}}$, there exists a unique element $\delta \in \mathbb{F}^{\ast}_{2^{k}}$ such that $\delta^{2^{k}+1}=\lambda\lambda^{-1}_{0}$. Then we have $G_{\lambda}(x)={\rm Tr}^{n}_{1}(\lambda \omega x^{2^{k}+1})={\rm Tr}^{n}_{1}(\lambda \lambda^{-1}_{0}\lambda_{0}\omega x^{2^{k}+1})={\rm Tr}^{n}_{1}(\lambda_{0}\omega (\delta x)^{2^{k}+1})=G_{\lambda_{0}}(\delta x)$. By similar arguments as in Proposition \ref{Li} and the fact that $G^{\ast}_{\lambda_{0}}(x)$ is self-dual,  we have $G^{\ast}_{\lambda}(x)=G^{\ast}_{\lambda_{0}}(\delta^{-1}x)=G_{\lambda_{0}}(\delta^{-1}x).$ Note that for any $a, b\in \mathbb{F}_{2^{n}}$, $D_{a}D_{b}G^{\ast}_{\lambda}(x)=D_{a}D_{b}G_{\lambda_{0}}(\delta^{-1}x)={\rm Tr}^{n}_{1}(\lambda^{2}_{0}\lambda^{-1}\omega(a^{2^{k}}b+ab^{2^{k}})).$ Thus, if $ab^{2^{k}}\in \mathbb{F}_{2^{k}}$, then  $D_{a}D_{b}G^{\ast}_{\lambda}(x)=0$. We are done. \end{proof}

\begin{thm}\label{c} Let $n=4k$ with $k\geq 2$ and $\tau$ be an integer such that $1\leq \tau \leq k$. Let $\{u_{1},\ldots,u_{k}\}\subseteq \mathbb{F}^{\ast}_{2^{2k}}$ such that $u_{i}u^{2^{k}}_{j}\in \mathbb{F}^{\ast}_{2^{k}}$ for any $1\leq i<j\leq k$. Let $F(X_{1},X_{2},\ldots, X_{\tau})$ be any reduced polynomial in $\mathbb{F}_{2}[X_{1},X_{2},\ldots, X_{\tau}]$ with algebraic degree $d$. Then $H(x)={\rm Tr}^{n}_{k}(\omega x^{2^{k}+1})+F({\rm Tr}^{n}_{1}(u_{i_{1}}x), \ldots,{\rm Tr}^{n}_{1}(u_{i_{\tau}}x))$ is a vectorial bent $(n,k)$-function, where $\omega$ is a generator of the cyclic group  $U=\{x\in \mathbb{F}_{2^{2k}}~|~x^{2^{k}+1}=1\}$, and $\{i_{1},\ldots, i_{\tau} \} \subseteq \{1,\ldots,k\}$. Furthermore, if $u_{1},\ldots,u_{k}$ are linearly independent over $\mathbb{F}_{2}$ and  $d\geq 2$, then $H(x)$ has algebraic degree $d$. \end{thm}
\begin{rmk} Let the notations be defined in Theorem \ref{c}. Let $\{v_{1},\ldots,v_{k} \}$ be a basis of $\mathbb{F}_{2^{k}}$ over $\mathbb{F}_{2}$ and $\varsigma$ be any element in $U\backslash\{1\}$. Set $u_{i}:=v_{i}\varsigma$, $i=1,\ldots k$, it is clear that $u_{i}u^{2^{k}}_{j}\in \mathbb{F}_{2^{k}}$ for any $1\leq i<j\leq k$. This means that we have many choices of $u^{\prime}_{i}s$ in Theorem \ref{c} to get the desired vectorial bent function.
\end{rmk}

\begin{coro}\label{c-vm} Assuming conditions of Theorem \ref{c}, and $t$ a positive integer. Let $F_{i}(X_{1},\ldots, X_{k})$, $i=1,
\ldots,t$, be any reduced polynomials in $ \mathbb{F}_{2}[X_{1},\ldots,X_{k}]$. Set $f_{i}(x):=F_{i}({\rm Tr}^{n}_{1}(u_{1}x)), \ldots,{\rm Tr}^{n}_{1}(u_{k}x))$ for each $1\leq i\leq t$. Then $\widehat{H}(x)=(x^{2^{k}+1},f_{1}(x),\ldots,f_{t}(x))$ is a vectorial plateaued $(n,k+t)$-function if and only if the $(n,t)$-function $(f_{1}(x),\ldots,f_{t}(x))$ is vectorial plateaued.
\end{coro}
\begin{proof} This can be seen from Corollary \ref{vecplateaued} and Theorem \ref{c}. \end{proof}
\begin{rmk} With similar arguments as in Corollary \ref{c-vm}, one has that the function  $\widehat{H}(x)$ described above is a vectorial function with maximal number of bent components.
\end{rmk}

\section{Concluding Remarks}

In this paper, we proposed a generic method to construct vectorial bent (plateaued) functions via the second-order derivatives, and obtained (at least) three infinite families of vectorial bent (plateaued) from the following three classes of $(n,n)$-functions: $G_{1}(x)=x^{2^{k}+1}$ with $n=2k$; $G_{2}(x)=\sum\limits^{2^{r}-1}_{i=1}x^{(i2^{k-r}+1)(2^{k}-1)+1}$ with $n=2k$, $1<r<k$ and ${\rm gcd}(r,k)=1$; $G_{3}(x)=x^{2^{k}+1}$ with $n=4k$. In particular, the generic construction can produce vectorial bent functions with high algebraic degrees and vectorial plateaued functions having the maximal number of bent components.

\end{document}